\documentclass[11pt]{article}

\usepackage{graphics,amsmath,amssymb,amsthm,accents}
\usepackage{tipa}
\usepackage{epic}
\usepackage{amsfonts}
\usepackage{accents}
\usepackage{texdraw}
\usepackage{color}
\usepackage{eucal}
\usepackage{cite}
\usepackage{bm}
\newtheorem{Proposition}{Proposition}
\numberwithin{equation}{section}

\def \tyb#1{\hbox{\tiny{[{\it{#1}}]}}}
\def \ty#1{\hbox{\tiny{{\it{#1}}}}}

\def \S#1{S^{(#1)}}

\DeclareMathAccent{\wtilde}{\mathord}{largesymbols}{"65}
\DeclareMathAccent{\what}{\mathord}{largesymbols}{"62}

\def\m@th{\mathsurround=0pt}
\mathchardef\bracell="0365
\def\upbrall{$\m@th\bracell$}
\def\undertilde#1{\mathop{\vtop{\ialign{##\crcr
    $\hfil\displaystyle{#1}\hfil$\crcr
     \noalign
     {\kern1.5pt\nointerlineskip}
     \upbrall\crcr\noalign{\kern1pt
   }}}}\limits}

\def\underhat#1{\mathop{\vtop{\ialign{##\crcr
    $\hfil\displaystyle{#1}\hfil$\crcr
     \noalign
     {\kern1.5pt\nointerlineskip}
     \upbrall\crcr\noalign{\kern1pt
   }}}}\limits}

\newcommand{\sg}{\sigma}

\newcommand{\nn}{\nonumber}

\newcommand{\bA}{\boldsymbol{A}}
\newcommand{\bB}{\boldsymbol{B}}
\newcommand{\bC}{\boldsymbol{C}}
\newcommand{\bI}{\boldsymbol{I}}
\newcommand{\bK}{\boldsymbol{K}}

\newcommand{\bM}{\boldsymbol{M}}
\newcommand{\bS}{\boldsymbol{S}}

\newcommand{\br}{\boldsymbol{r}}
\newcommand{\bs}{\boldsymbol{s}}
\newcommand{\bu}{\boldsymbol{u}}

\newcommand{\bT}{\boldsymbol{T}}
\newcommand{\bF}{\boldsymbol{F}}
\newcommand{\bG}{\boldsymbol{G}}
\newcommand{\bH}{\boldsymbol{H}}
\newcommand{\bX}{\boldsymbol{X}}
\newcommand{\bY}{\boldsymbol{Y}}
\newcommand{\bZ}{\boldsymbol{Z}}
\newcommand{\Ga}{\boldsymbol{\Gamma}}

\newcommand{\st}{\hbox{\tiny\it{T}}}

\begin{document}

\title{{The Sylvester equation and integrable equations: I. \\
The Korteweg-de Vries system and sine-Gordon equation}}

\author{Dan-dan Xu$^1$,~ Da-jun Zhang$^1$\footnote{Corresponding author. Email: djzhang@staff.shu.edu.cn},~
Song-lin Zhao$^2$\\
~~\\
{\small \it ${}^{1}$ Department of Mathematics, Shanghai University, Shanghai 200444, P.R. China}\\
{\small \it  ${}^{2}$ Department of Mathematics, Zhejiang University of Technology, Hangzhou 310023, P.R. China}
}

\maketitle

\begin{abstract}

The paper is to reveal the direct links between the well known Sylvester equation in matrix theory
and some integrable systems.
Using the Sylvester equation $\boldsymbol{K} \boldsymbol{M}+\boldsymbol{M} \boldsymbol{K}=\boldsymbol{r}\, \boldsymbol{s}^{T}$ 
we introduce a scalar function
$S^{(i,j)}=\boldsymbol{s}^{T}\, \boldsymbol{K}^j(\boldsymbol{I}+\boldsymbol{M})^{-1}\boldsymbol{K}^i\boldsymbol{r}$
which is defined as same as in discrete case.
$S^{(i,j)}$ satisfy some recurrence relations which can be viewed as discrete equations
and play indispensable roles in deriving continuous integrable equations.
By imposing dispersion relations on  $\boldsymbol{r}$ and $\boldsymbol{s}$,
we find the Korteweg-de Vries  equation, modified Korteweg-de Vries  equation,
Schwarzian Korteweg-de Vries equation and sine-Gordon equation
can be expressed by some discrete equations of $S^{(i,j)}$ defined on certain points.
Some special matrices are used to solve the Sylvester equation
and prove symmetry property $S^{(i,j)}=S^{(i,j)}$.
The solution $\boldsymbol{M}$ provides $\tau$ function by $\tau=|\boldsymbol{I}+\boldsymbol{M}|$.
We hope our results can not only unify the Cauchy matrix approach in both continuous and discrete cases,
but also bring more links for integrable systems and variety of areas where the
Sylvester equation appears frequently.

\vskip 8pt \noindent {\bf Keywords:} The Sylvester equation, integrable systems, Cauchy matrix approach, solutions
\noindent {\bf PACS:}\quad  02.30.Ik, 05.45.Yv, 02.10.Yn


\end{abstract}

\maketitle

%

\section{Introduction}
\label{sec-1}

The Sylvester equation
\begin{equation}
\bA\bM-\bM\bB=\bC
\label{SE}
\end{equation}
is one of the most well-known matrix equations.
It appears frequently in many areas of applied mathematics and plays  a central role in particular in systems and control theory,
signal processing, filtering, model reduction, image restoration, and so on.
J. Sylvester is the first mathematician who introduced the term ``Matrix'' to name a matrix of the present form.
In the equation \eqref{SE} $\bA, \bB$ and $\bC$ are known matrices and $\bM$ is the unknown matrix.
It is also known as the Rosenblum  equation in operator theory.
We refer the reader to the elegant survey \cite{BR-BLMS-1997} by  Bhatia and Rosenthal and the references therein
for a history of the equation and many interesting and important theoretical results.

In the paper we will investigate the role of the Sylvester equation \eqref{SE} in the field of integrable systems.
Integrable systems mean the exactly solvable nonlinear partial differential (and difference) equations with regular solution structures
(e.g. $N$-soliton solutions, etc.).
From the glance it is hard to relate the Sylvester equation \eqref{SE} and an integrable system together.

In fact, the Sylvester equation \eqref{SE} appears in many contexts of integrable systems,
for example, the Cauchy matrix approach,
the operator method
and
the bidifferential calculus approach.
The Cauchy matrix approach, as a systematic method for constructing discrete integrable equations together with their solutions,
was first proposed by Nijhoff and his collaborators \cite{NAH-2009-JPA,N-2004-math}.
In this method, the discrete plain wave factors $r_i=\bigl(\frac{a-k_i}{a+k_i}\bigr)^n\bigl(\frac{b-k_i}{b+k_i}\bigr)^m r^{(0)}_i$
satisfy a Sylvester equation
\begin{equation}
\bK\bM+ \bM \bK = \br\, \bs^{\st},
\label{SE-1}
\end{equation}
where
\begin{align*}
& \bK=\mathrm{Diag}(k_1,k_2,\cdots,k_N),~~\bM=(M_{i,j})_{N\times N},~M_{i,j}=\frac{r_i s_j}{k_i+k_j},~~\br=(r_1,r_2,\cdots,r_N)^T,
\end{align*}
and $\bs=(s_1,s_2,\cdots,s_N)^T$ and 
$(r^{(0)}_1,r^{(0)}_2,\cdots,r^{(0)}_N)^T$ are  constant vectors.
Then,  scalar functions
$
\S{i,j}=\bs^{\st}\,\bK^j(\bI+\bM)^{-1}\bK^i\br
$
obey some recurrence relations,
and among them there are closed forms which give rise to discrete integrable equations.
One can view \eqref{SE-1} as a Sylvester equation containing an unknown $\bM$ and a generic constant matrix $\bK$.
Then, more general solutions can be derived for discrete systems \cite{ZZ-SAM-2013}.
The operator method (or trace method in scalar case), based on Marchenko's work \cite{M-book-1987},
was first proposed by Aden and Carl \cite{AC-1996-JMP},
and developed by Schiebold and her collaborators \cite{S-PD-1998,CS-Non-1999,CS-DMV-2000}.
In this method, suitable dispersion relations are imposed on $\bM$ or $\Ga=\bI+\bM$
(e.g., $\bM_x=\bK\bM,~\bM_t=\bK^3 \bM$)
and solutions of nonlinear partial equations are expressed in the form of logarithmic derivative $\Ga^{-1}\Ga_x$
or its trace.
This method needs lengthy verification of solutions. For the review of this method one can see \cite{CS-Non-1999,CS-DMV-2000}.
This method relies on the Sylvester equation \eqref{SE} with $\bC$ of rank one so as to get needed trace property.
It is remarkable that in \cite{S-LAA-2010} Schiebold collected examples of the correspondence
of the Sylvester equation \eqref{SE} and solutions of some integrable systems.
She also derived explicit solutions of \eqref{SE} for the case $\bA,\bB$ having Jordan block canonical forms
and extended formulae of Cauchy-type determinants.
The operator method for scalar case is also  viewed as a trace method in \cite{B-Non-2000}
where we note that operator solutions to the Marchenko's integral equation were given.
In \cite{AM-IP-2006,ADM-IP-2007,ADM-JMP-2010,D-TMP-2011}
solutions of the Gel'fand-Levitan-Marchenko (GLM) equation are expressed  via a triplet $(\bA,\bB,\mathbf{C})$
where matrix $\bA$ and vectors $\bB$ and $\mathbf{C}$ satisfy some Sylvester equations.
The bidifferential calculus approach (see the review paper \cite{DM-DCDS-2009} and the references therein).
In this approach, integrable equations are derived by introducing  bidifferential operators $\mathrm{d}$ and $\bar{\mathrm{d}}$ into graded algebras.
Solutions of the obtained integrable equations can be  parameterised in terms of some matrices which satisfy Sylvester equations,
(see \cite{DM-IP-2010,DM-SIGMA-2010} as more examples).
One more example is given in \cite{HL-PLA-2001},
where solutions to the Kadomtsev-Petviashvili (KP) equation are given in the form $w=p^T \mathbf{C}^{-1}q$ and the Sylvester equation also appears
in the solving procedure.

In this paper we would like to reveal more primary links between the  Sylvester equation and integrable equations.
We will see that the  Sylvester equation plays a basic role in the sense of constructing integrable equations and their solutions.
This is already realized in \cite{S-LAA-2010} for continuous case and
in Refs.\cite{NZZ-2013,Z-INI-2013} for discrete integrable equations.
In this paper, motivated by the discrete Cauchy matrix approach  and the understanding discrete dispersion relation,
we will impose the dispersion relation on $\br$ or $\bs$ (or both)
rather than on $\bM$ in \eqref{SE-1},
and focus on the evolution of the scalar function $\S{i,j}$.
We can then not only unify the Cauchy matrix approach in discrete and continuous cases,
but also have chances to find more links between discrete systems and continuous ones.
In this paper, we will start from a Sylvester equation (see \eqref{SE-2}) and
examine the links between the Sylvester equation and some continuous integrable equations,
such as the Korteweg-de Vries (KdV) equation, modified Korteweg-de Vries (mKdV) equation,
Schwarzian Korteweg-de Vries (SKdV) equation and sine-Gordon equation.
We will see that $\S{i,j}$ (defined as \eqref{Sij}) compose an infinite symmetric matrix
and $\S{i,j}$ obey some recurrence relations as they do in discrete case.
These recurrence relations can be viewed as discrete equations of $\S{i,j}$ with
discrete independent variables $i,j$.
It turns out that the continuous equations considered in our paper (usually are their potential forms)
are equivalent to some discrete equations of $\S{i,j}$ sitting on some special points.

The paper is organized as follows.
In Sec.\ref{sec-2} we introduce the Sylvester equation and some related properties including solvability,
recurrence relations and symmetric property of $\S{i,j}$.
In Sec.\ref{sec-3} we impose evolution on the elements $\br, \bs$ in the Sylvester equation \eqref{SE-2}
and derive the KdV equation, mKdV equation and SKdV equation.
Sec.\ref{sec-4} derives the sine-Gordon equation.
In Sec.\ref{sec-5} we discuss the relations between $\S{i.j}$
and $\tau$ functions, the trace in the operator method, etc.
Sec.\ref{sec-6} contains conclusions and further discussions.
In addition, we have an Appendix consisting of 5 sections
as a compensation of the paper.

\section{The Sylvester equation and some related properties}\label{sec-2}

\subsection{Solvability}

For the solution of the Sylvester equation \eqref{SE}, a well known result,
which was proved by Sylvester\cite{Sylvester-1884}, is
\begin{Proposition}\label{prop-1}
Let us denote the eigenvalue sets of $\bA$ and $\bB$ by $\mathcal{E}(\bA)$ and $\mathcal{E}(\bB)$, respectively.
For the known $\bA, \bB$ and $\bC$, the Sylvester equation \eqref{SE} has a unique solution $\bM$ if and only if
$\mathcal{E}(\bA)\bigcap \mathcal{E}(\bB)=\varnothing$.
\end{Proposition}

When the eigenvalues of $\bA$ and $\bB$ satisfy certain conditions, the solution of the Sylvester equation \eqref{SE}
can be expressed via series or integration. (See Ref.\cite{BR-BLMS-1997} and the references therein.)
\begin{Proposition}\label{prop-1-1}
When $\mathcal{E}(A)\subset \{z: |z|>\rho\}$ and $\mathcal{E}(B)\subset \{z:|z|<\rho\}$ for some $\rho>0$,
then the solution of the Sylvester equation \eqref{SE} is
\[\bM =\sum^{\infty}_{j=0}\bA^{-j-1}\bC \bB^{j}.\]
When $\mathcal{E}(\bA)$ and $\mathcal{E}(\bB)$ are contained in the open right half plane and the open left half plane, respectively,
then the solution of the Sylvester equation \eqref{SE} is
\[\bM =\int^{\infty}_{0}e^{-t\bA}\bC e^{t \bB} dt.\]
\end{Proposition}

The Sylvester equation of our interest in this paper is
\begin{equation}
\bK \bM+\bM\bK=\br\, \bs^{\st},
\label{SE-2}
\end{equation}
where $\bK, \bM \in \mathbb{C}_{N\times N}$, $\br=(r_1,r_2,\cdots,r_N)^T$ and $\bs=(s_1,s_2,\cdots,s_N)^{T}$.
This corresponds to $\bA=-\bB$ and $\bC$ being of rank 1 in \eqref{SE}.
In light of proposition \ref{prop-1}, the equation \eqref{SE-2} is solvable when $0\notin \mathcal{E}(\bK)$
or equivalently, $|\bK|\neq 0$.
Under such a condition, an explicit form (neither series nor integration) of the solution of the Sylvester equation \eqref{SE-2}
were presented in  \cite{ZZ-SAM-2013,S-LAA-2010}
The solution is obtained by factorizing $\bM$ into $\bM=\bF \bG \bH$ where $\bF, \bG, \bH$ are some $N\times N$ matrices.
For the completeness of the paper, we would like to give solving procedure and results in Appendix \ref{A:2}.

\subsection{Infinite matrix $\bS$}

\subsubsection{Recurrence relation of $\S{i,j}$}

By the Sylvester equation \eqref{SE} we introduce an $\infty\times\infty$ matrix $\bS=(S^{(i,j)})_{\infty\times\infty}$ where
the element $S^{(i,j)}$ is defined as (cf. \cite{NAH-2009-JPA,N-2004-math})
\begin{equation}
S^{(i,j)}=\bs^{\st}\,\bK^j(\bI+\bM)^{-1}\bK^i\br, ~~i,j\in \mathbb{Z},
\label{Sij}
\end{equation}
where $\bI$ is the $N$th-order unit matrix.
In the paper we call $\S{i,j}$ a master function because
it is used to generate integrable equations.
Let us discuss some properties of the matrix $\bS$ and the elements $S^{(i,j)}$.

\begin{Proposition}\label{prop-2}
For the master function $S^{(i,j)}$ defined by \eqref{Sij} with $\bM, \bK, \br,\bs$ satisfying the Sylvester equation \eqref{SE-2},
we have the following relation,
\begin{equation}
\S{i,j+2k}=\S{i+2k,j}-\sum^{2k-1}_{l=0}(-1)^l\S{2k-1-l,j}\S{i,l},~~~(k=1,2,\cdots).
\label{Sij-k+}
\end{equation}
Particularly, when $k=1$ we have (see also equation (2.16) in \cite{NAH-2009-JPA})
\begin{equation}
S^{(i,j+2)}=S^{(i+2,j)}-S^{(i,0)}S^{(1,j)}+S^{(i,1)}S^{(0,j)}.
\label{Sij-k=1+}
\end{equation}
\end{Proposition}

\begin{proof}
First, from the Sylvester equation \eqref{SE-2} we have the following relation
\begin{equation}
\bK^s \bM -(-1)^{s} \bM \bK^s=\sum^{s-1}_{j=0}(-1)^j\bK^{s-1-j} \br \bs^{\st} \bK^j,~~~(s=1,2,\cdots).
\label{MK-rec}
\end{equation}
In fact, obviously, when $s=1$, \eqref{MK-rec} is nothing but the Sylvester equation \eqref{SE-2} itself.
Now left-multiplying $\bK$ on \eqref{SE-2} yields
\[\bK^2 \bM +\bK \bM \bK= \bK  \br \bs^{\st}.\]
Using the Sylvester equation \eqref{SE-2} we replace the term $\bK \bM$ with $-\bM \bK+\br \bs^{\st}$
we have
\[
\bK^2 \bM - \bM \bK^2 = - \br \bs^{\st}\bK +\bK \br \bs^{\st} = \sum^{1}_{j=0}(-1)^j\bK^{1-j} \br \bs^{\st} \bK^j,\]
which is $s=2$ in \eqref{MK-rec}.
Next, repeating the same procedure or using mathematical induction we can reach \eqref{MK-rec}.

To prove the relation \eqref{Sij-k+}, we introduce an auxiliary vector
\begin{equation}
\bu^{(i)}=(\bI+\bM)^{-1} \bK^i \br,~~ i\in \mathbb{Z}.
\label{ui}
\end{equation}
From this we have
\[\bK^s \bu^{(i)} + \bK^s \bM \bu^{(i)}=\bK^{s+i} \br.\]
Then, taking $s=2k$ and using the relation \eqref{MK-rec}, it is not difficult to find
\[ (\bI+\bM)\bK^{2k} \bu^{(i)}=\bK^{2k+i} \br
-\sum^{2k-1}_{l=0} (-1)^{l}\bK^{2k-1-l} \br \bs^{\st}  \bK^l \bu^{(i)},
\]
which is, then, left-multiplied by $\bs^{\st}\bK^j (\bI+\bM)^{-1}$ yields the relation \eqref{Sij-k+}.
\end{proof}

In a similar way, we can have the following similar property.
\begin{Proposition}\label{prop-3}
For the scalar function $S^{(i,j)}$ defined by \eqref{Sij} with $\bM, \bK, \br,\bs$ satisfying the Sylvester equation \eqref{SE-2},
we have the following relation,
\begin{equation}
\S{i,j-2k}=\S{i-2k,j}+\sum^{2k-1}_{l=0}(-1)^l\S{i,-2k+l}\S{-1-l,j},~~~(k=1,2,\cdots).
\label{Sij-k-}
\end{equation}
When $k=1$ one has
\begin{equation}
S^{(i,j-2)}=S^{(i-2,j)}+S^{(i,-2)}S^{(-1,j)}-S^{(i,-1)}S^{(-2,j)}.
\label{Sij-k=1-}
\end{equation}
\end{Proposition}

\subsubsection{Invariance of $\S{i,j}$ }\label{sec:2.2.2}

Let us suppose $\bK_1$ is the matrix that is similar to $\bK$ under the transform matrix $\bT$, i.e.,
\begin{subequations}
\label{trans-sim}
\begin{equation}
\bK_1=\bT \bK \bT^{-1}.
\label{K1-K}
\end{equation}
We also denote
\begin{equation}
\bM_1=\bT \bM \bT^{-1},~~\br_1=\bT \br,~~\bs_1^{\st}=\bs^{\st} \bT^{-1}.
\label{Mrs-1}
\end{equation}
\end{subequations}
It is easy to verify that
\begin{subequations}
\begin{equation}
\bM_1 \bK_1 + \bK_1\bM_1= \br_1 \,\bs^{\st}_1.
\label{SE-co-inv}
\end{equation}
and
\begin{equation}
S^{(i,j)} = \bs^{\st} \bK^j(\bI+ \bM)^{-1} \bK^i \br
= \bs_1^{\st} \bK_1^j(\bI+ \bM_1)^{-1}\bK_1^i \br_1.
\end{equation}
\end{subequations}
Therefore we can say that $\S{i,j}$ is invariant under the similar transformation \eqref{trans-sim}.

\subsubsection{Symmetry property of $\S{i,j}$}

\begin{Proposition}\label{prop-4}
Suppose that $\bK, \bM, \br, \bs$ satisfy the Sylvester equation \eqref{SE-2} and $|\bK| \neq 0$.
Then the scalar elements $\S{i,j}$ defined by \eqref{Sij} satisfy the symmetry property
\begin{equation}
\S{i,j}=\S{j,i},
\label{Sij=Sji}
\end{equation}
i.e., the infinite matrix $\bS$ is symmetric.
\end{Proposition}

The proof will be given in Appendix \ref{A:3}.

\section{The KdV system}\label{sec-3}

\subsection{Evolution of $\bM$}\label{sec:3.1}

Now we suppose that $\br, \bs, \bM$ are functions of $(x,t)$ while $\bK$ is a non-trivial constant matrix,
and the evolution of $\br$ and $\bs$ are formulated by
\begin{subequations}
\begin{align}
& \br_x=\bK \br,~~\bs_x=\bK^T \bs, \label{evo-rs-x} \\
& \br_t=4\bK^3 \br,~~\bs_t=4(\bK^T)^3\bs. \label{evo-rs-t}
\end{align}
\label{evo-rs}
\end{subequations}
Taking the derivative of the Sylvester equation \eqref{SE-2} w.r.t. $x$ and making using of \eqref{evo-rs-x} we have
\[
\bK \bM_x+ \bM_x\bK=\br_x \,\bs^{\st}+ \br \bs^{\st}_x=\bK \br \bs^{\st}+ \br \bs^{\st} \bK,
\]
i.e.
\[\bK (\bM_x-\br \bs^{\st})+ (\bM_x-\br \bs^{\st} )\bK=0. \]
This immediately yields the relation
\begin{equation}
\bM_x=\br \bs^{\st},
\label{evo-Mx}
\end{equation}
by virtue of proposition \ref{prop-1}.
Obviously, the above relation is also written as
\begin{equation}
\bM_x=\bK \bM+ \bM\bK
\label{Mx-1}
\end{equation}
if using the Sylvester equation \eqref{SE-2}.

In a similar way, for the time evolution of $\bM$, we have
\[
\bK \bM_t+\bM_t \bK =\br_t \,\bs^{\st}+ \br \bs^{\st}_t=4(\bK^3 \br \bs^{\st}+ \br \bs^{\st} \bK^3).
\]
Then, replacing $\br \bs^{\st}$ with $\bK \bM+ \bM\bK$ and making using of  proposition \ref{prop-1}
we find
\begin{equation}
\bM_t=4(\bK^3 \bM+ \bM\bK^3).
\label{Mt3}
\end{equation}

Actually, it can be proved that if
\begin{equation}
\br_{t_j}=\gamma \bK^j \br,~~\bs_{t_j}=\gamma (\bK^T)^j\bs , ~~ j\in \mathbb{Z}
\label{evo-rs-tj}
\end{equation}
with constant $\gamma$, then the evolution of $\bM$ is formulated by
\begin{equation}
\bM_{t_j}=\gamma(\bK^j \bM+ \bM\bK^j).
\label{Mtj}
\end{equation}

Go back to the relation \eqref{Mt3}. Corresponding to the expression \eqref{evo-Mx}, \eqref{Mt3}
can be alternatively expressed as
\begin{equation}
\bM_t=4(\bK^2\br\bs^{\st}-\bK\br\bs^{\st}\bK+\br\bs^{\st}\bK^2).
\label{evo-Mt3}
\end{equation}

With the evolution formulas \eqref{evo-rs}, \eqref{evo-Mx} and \eqref{evo-Mt3},
we can derive the evolution of $\S{i,j}$.

\subsection{Evolution of $\S{i,j}$}

Let us recall the auxiliary vector $\bu^{(i)}$ defined in \eqref{ui}, i.e.
\begin{equation}
\bu^{(i)}=(\bI+\bM)^{-1}\bK^i\br
\label{ui-2}
\end{equation}
and the connection with $\S{i,j}$:
\begin{equation}
S^{(i,j)}=\bs^{\st}\bK^j(\bI+\bM)^{-1}\bK^i\br=\bs^{\st} \bK^j\bu^{(i)}. \label{Sij-ui}
\end{equation}
Taking $x$-derivative on \eqref{ui-2}  we have
\[
\bM_x\bu^{(i)}+(\bI+\bM)\bu^{(i)}_x=\bK^{i}\br_x=\bK^{i+1}\br.
\]
and further,
\[
(\bI+\bM)\bu^{(i)}_x=\bK^{i+1}\br-\br \bs^{\st} \bu^{(i)},
\]
where we have substituted \eqref{evo-Mx}.
Then, left-multiplied the above relation by $(\bI+\bM)^{-1}$ and also using the relation \eqref{Sij-ui}
yield
\begin{equation}
\bu^{(i)}_x=\bu^{(i+1)}-S^{(i,0)}\bu^{(0)},
\label{evo-uix}
\end{equation}
which is viewed as evolution of $\bu^{(i)}$ in $x$-direction.
In a similar way we derive the evolution of $\bu^{(i)}$ in $t$-direction:
\begin{eqnarray}
\bu^{(i)}_{t}=4(\bu^{(i+3)}-S^{(i,0)}\bu^{(2)}-S^{(i,2)}\bu^{(0)}+S^{(i,1)}\bu^{(1)}). \label{evo-uit}
\end{eqnarray}
Now, noting that the connection \eqref{Sij-ui} between $\bu^{(i)}$ and $\S{i,j}$,
we  left-multiply $\bs^{\st}\bK^j$ on relations \eqref{evo-uix} and \eqref{evo-uit}, respectively.
After some calculation we obtain the following evolution for $S^{(i,j)}$:
\begin{subequations}
\label{evo-Sij}
\begin{align}
& S^{(i,j)}_{x}=S^{(i+1,j)}+S^{(i,j+1)}-S^{(i,0)}S^{(0,j)}, \label{evo-Sijx} \\
& S^{(i,j)}_t=4(S^{(i+3,j)}+S^{(i,j+3)}+S^{(i,1)}S^{(1,j)}-S^{(i,0)}S^{(2,j)}-S^{(i,2)}S^{(0,j)}). \label{evo-Sijt}
\end{align}
One can repeatedly use \eqref{evo-Sijx} and get higher-order derivatives\footnote{These can be easily derived by means of
computer algebra, e.g. {\it Mathematica}.}
 of $S^{(i,j)}$ w.r.t $x$. Here we just
list $S^{(i,j)}_{xx}$ and $S^{(i,j)}_{xxx}$, which read
\begin{align}
S^{(i,j)}_{xx}=& S^{(i+2,j)}+S^{(i,j+2)}-2S^{(i+1,0)}S^{(0,j)}-2S^{(i,0)}S^{(0,j+1)}\nn\\
& +2S^{(i+1,j+1)}-S^{(i,0)}S^{(1,j)}-S^{(i,1)}S^{(0,j)}+2S^{(i,0)}S^{(0,0)}S^{(0,j)}, \label{evo-Sijxx}\\
S^{(i,j)}_{xxx}=& S^{(i+3,j)}+S^{(i,j+3)}+3S^{(i+2,j+1)}+3S^{(i+1,j+2)}-3S^{(i+2,0)}S^{(0,j)}\nn\\
&  -3S^{(i,0)}S^{(0,j+2)}-6S^{(i+1,0)}S^{(0,j+1)}-3S^{(i+1,0)}S^{(1,j)}-3S^{(i,1)}S^{(0,j+1)}\nn\\
&  -S^{(i,2)}S^{(0,j)}-S^{(i,0)}S^{(2,j)}-3S^{(i+1,1)}S^{(0,j)}-3S^{(i,0)}S^{(1,j+1)}\nn\\
&  +6S^{(i+1,0)}S^{(0,0)}S^{(0,j)}+6S^{(i,0)}S^{(0,0)}S^{(0,j+1)}-2S^{(i,1)}S^{(1,j)}\nn\\
&  +3S^{(i,0)}S^{(0,0)}S^{(1,j)}+6S^{(i,0)}S^{(1,0)}S^{(0,j)}+3S^{(i,1)}S^{(0,0)}S^{(0,j)}\nn\\
&  -6S^{(i,0)}{S^{(0,0)}}^2 S^{(0,j)}. \label{evo-Sijxxx}
\end{align}
\end{subequations}

\subsection{Nonlinear evolution equations}

\subsubsection{The KdV equation}

Let us define
\begin{equation}
u=\S{0,0},
\label{u}
\end{equation}
i.e. $i=j=0$ for $\S{i,j}$.
In this case, with symmetric property $\S{i,j}=\S{j,i}$ the evolution relation \eqref{evo-Sij} reduces to
\begin{subequations}
\begin{align}
& u_{x}=2S^{(1,0)}-u^{2}, \label{ux} \\
& u_{t}=4(2S^{(3,0)}-2uS^{(2,0)}+S^{(1,0)^{2}}),\\
& u_{xxx}=2S^{(3,0)}+6S^{(2,1)}-8uS^{(2,0)}-14S^{(1,0)^{2}}-6uS^{(1,1)}+24u^2S^{(1,0)}-6u^4.
\end{align}
\end{subequations}
Then by direct substitution  we find
\begin{eqnarray}
u_{t}-u_{xxx}-6u_{x}^{2}=6(S^{(3,0)}-S^{(2,1)}-S^{(1,0)^{2}}+\S{0,0}S^{(1,1)}).
\end{eqnarray}
The right hand side is nothing but the recurrence relation  \eqref{Sij-k=1+} with $(i,j)=(1,0)$.
Thus we have
\begin{eqnarray}
u_{t}-u_{xxx}-6u_{x}^{2}=0, \label{pKdV}
\end{eqnarray}
which is the potential KdV equation.
Taking $w=2u_x$ yields the KdV equation
\begin{eqnarray}
w_{t}-w_{xxx}-6ww_{x}=0, \label{KdV}
\end{eqnarray}
of which the solution is given by
\begin{equation}
w=2(\bs^{\st}(\bI+\bM)^{-1}\br)_x.
\end{equation}

\subsubsection{The mKdV equation}

Let introduce
\begin{equation}
v=S^{(-1,0)}-1. \end{equation}
In this turn, with $(i,j)=(-1,0)$ in \eqref{evo-Sij} and using $S^{(i,j)}=S^{(j,i)}$ one finds
\begin{subequations}
\begin{align}
v_{x}=& S^{(-1,1)}-uv,\\
v_{t}=& 4(S^{(-1,3)}-vS^{(0,2)}+S^{(-1,1)}S^{(0,1)}-S^{(-1,2)}u),\\
v_{xx}=& -3vS^{(0,1)}+S^{(-1,2)}+2{u}^{2}v-uS^{(-1,1)}, \label{vxx} \\
v_{xxx}=& S^{(-1,3)}-4vS^{(0,2)}-3vS^{(1,1)}+15uvS^{(0,1)}\nn \\
        & -6{u}^{3}v-5S^{(0,1)}S^{(-1,1)}+3{u}^{2}S^{(-1,1)}-uS^{(-1,2)}.
\end{align}
\end{subequations}
Further we find
\begin{align*}
& v_{t}-v_{xxx}+3\frac{v_{x}v_{xx}}{v} \\
 =&
 3(S^{(-1,3)}+S^{(1,1)}v-S^{(-1,1)}S^{(0,1)})+3(\frac{\S{-1,1}}{v}-2u)(S^{(-1,2)}+vS^{(0,1)}-uS^{(-1,1)}),
\end{align*}
where the r.h.s. vanishes in the light of the recurrence relation \eqref{Sij-k=1+} with $(i,j)=(-1,0)$
and $(-1,1)$. Thus we have
\begin{eqnarray}
v_{t}-v_{xxx}+3\frac{v_{x}v_{xx}}{v}=0. \label{pmKdV-0}
\end{eqnarray}
This can be viewed as the potential mKdV equation.
In fact, by introducing
\begin{equation}
\nu= \ln v
\end{equation}
we can write \eqref{pmKdV-0} as
\begin{eqnarray}
\nu_{t}-\nu_{xxx}+2(\nu_x)^3=0, \label{pmKdV-01}
\end{eqnarray}
which is the familiar form of the potential mKdV equation.
Then, taking $\mu=\nu_x$ we arrive at the mKdV equation
\begin{eqnarray}
\mu_{t}-\mu_{xxx}+6\mu^2\mu_x=0, \label{mKdV}
\end{eqnarray}
of which the solution is given by
\begin{equation}
\mu=\partial_x \ln (\bs^{\st}\bK^{-1}(\bI+\bM)^{-1}\br-1).
\end{equation}

As for the connection with the KdV equation,
noting that, with the definition of $u$ and $v$, the recurrence relation \eqref{Sij-k=1+} with $(i,j)=(-1,0)$ reads
\begin{eqnarray}
S^{(-1,2)}=-vS^{(0,1)}+uS^{(-1,1)}, \label{S-12}
\end{eqnarray}
from \eqref{vxx} we find
\begin{eqnarray}
v_{xx}=(-4S^{(0,1)}+2u^{2})v=-2 u_x v. \label{vxx-s}
\end{eqnarray}
This implies
\begin{eqnarray}
-2u_x=\frac{v_{xx}}{v}=\left(\frac{v_x}{v}\right)_x+\left(\frac{v_x}{v}\right)^2, \label{MT-PKmK}
\end{eqnarray}
i.e.
\begin{eqnarray}
-w=\mu_x+\mu^2,
\end{eqnarray}
which is the  Miura transformation between the mKdV equation \eqref{mKdV} and the KdV equation \eqref{KdV}.

\subsubsection{The SKdV equation}

Let us examine the equation related to $\S{-1,-1}$. We introduce
\begin{eqnarray}
z=S^{(-1,-1)}-x.
\label{z}
\end{eqnarray}
Setting $i=j=-1$,  then \eqref{evo-Sij} gives
\begin{subequations}
\begin{eqnarray}
&& z_{x}=-v^{2}, \label{MT-MS} \\
&& z_{xx}=2uv^{2}-2vS^{(-1,1)},\\
&& z_{xxx}=-2vS^{(-1,2)}-2S^{(-1,1)^{2}}+6uvS^{(-1,1)}+6v^{2}S^{(0,1)}-6u^{2}v^{2},\\
&& z_{t}=4(S^{(-1,1)^2}-2vS^{(-1,2)}),
\end{eqnarray}
where we have made use of the definition of $u$ and $v$ and the symmetric property $S^{(i,j)}=S^{(j,i)}$.
\end{subequations}
By forwarding computation we find
\begin{eqnarray}
z_{t}-z_{xxx}+\frac{3}{2}\frac{z_{xx}^{2}}{z_{x}}=6v(uS^{(-1,1)}-S^{(-1,2)}-vS^{(0,1)}),
\end{eqnarray}
of which the r.h.s. obeys the recurrence relation \eqref{Sij-k=1+} with $(i,j)=(-1,0)$, i.e. \eqref{S-12}.
The SKdV equation reads
\begin{eqnarray}
z_{t}-z_{xxx}+\frac{3}{2}\frac{z_{xx}^{2}}{z_{x}}=0. \label{SKdV}
\end{eqnarray}
Relation \eqref{MT-MS} provides the following  Miura transformation between the potential mKdV equation and
SKdV equation,
\[\mu=\frac{z_{xx}}{2z_x}.\]

\vskip 6pt
Let us just give a brief summary of this section. We used the Sylvester equation \eqref{SE-2}
to define an infinite matrix $\bS$ with scalar elements $\S{i,j}$. By imposing evolution \eqref{evo-rs} on $\br, \bs$
(which are also viewed as dispersion relation), we find the potential KdV equation \eqref{pKdV},
the potential mKdV equation \eqref{pmKdV-01} and the SKdV equation \eqref{SKdV} are nothing but the recurrence relation \eqref{Sij-k=1+}
on certain points (i.e. with some choices of $(i,j)$).

\section{The sine-Gordon equation}\label{sec-4}

In this section we will consider the following dispersion relation (or evolution for $\br, \bs$):
\begin{subequations}
\begin{align}
& \br_x=\bK \br,~~ \bs_x=\bK^T \bs,\label{evo-rs-sg-x} \\
& \br_t=\frac{1}{4} \bK^{-1} \br,~~\bs_t=\frac{1}{4} (\bK^T)^{-1}\bs. \label{evo-rs-sg-t}
\end{align}
\label{evo-rs-sg}
\end{subequations}
By this we  derive the sine-Gordon equation.

Similar to the treatment in Sec.\ref{sec:3.1},
for the evolution of $\bM$, in addition to \eqref{evo-Mx}, corresponding to \eqref{evo-rs-sg-t}
we find
\begin{eqnarray}
\bM_t=\frac{1}{4} \bK^{-1}\br\bs^{\st}\bK^{-1}. \label{evo-Mt-sg}
\end{eqnarray}
For the auxiliary vector $\bu^{(i)}$ defined by \eqref{ui}, besides \eqref{evo-uix}, here we have
\begin{equation}
\bu_t^{(i)}=\frac{1}{4}(\bu^{(i-1)}-\bu^{(-1)} S^{(i,-1)}).
\label{evo-ut-sg}
\end{equation}
Then, for $\S{i,j}$ we find
\begin{subequations}\label{evo-Sij-sg}
\begin{align}
S^{(i,j)}_{x}=& S^{(i+1,j)}+S^{(i,j+1)}-S^{(i,0)}S^{(0,j)}, \label{evo-Sij-sg-x} \\
S_t^{(i,j)}=&\frac{1}{4}(S^{(i-1,j)}+S^{(i,j-1)}-S^{(i,-1)}S^{(-1,j)}),\label{evo-Sij-sg-t}\\
S_{xt}^{(i,j)}=&\frac{1}{4}[2S^{(i,j)}+S^{(i-1,j+1)}+S^{(i+1,j-1)}-S^{(i-1,0)}S^{(0,j)}-S^{(i,0)}S^{(0,j-1)}\nonumber\\
&-(S^{(i+1,-1)}+S^{(i,0)}-S^{(i,0)}S^{(0,-1)})S^{(-1,j)}\nonumber\\
&-S^{(i,-1)}(S^{(0,j)}+S^{(-1,j+1)}-S^{(-1,0)}S^{(0,j)})],\label{evo-Sij-sg-xt}
\end{align}
in which \eqref{evo-Sij-sg-x} is the same as before.
\end{subequations}

In order to derive the sine-Gordon equation, we employ the previous definition
\begin{eqnarray}
u=S^{(0,0)},~~v=S^{(-1,0)}-1,
\end{eqnarray}
while we use the dispersion relation \eqref{evo-rs-sg}. It then follows from \eqref{evo-Sij-sg} and the property
$\S{i,j}=\S{j,i}$ that
\begin{subequations}
\begin{eqnarray}
&& v_{x}=S^{(-1,1)}-uv,\\
&& v_{t}=\frac{1}{4}(-vS^{(-1,-1)}+S^{(-2,0)}),\\
&& v_{xt}=\frac{1}{4}(1+v^3+S^{(1,-2)}-S^{(1,-1)}S^{(-1,-1)}-uS^{(0,-2)}+uvS^{(-1,-1)}).
\end{eqnarray}
\end{subequations}

Then, by calculation we find
\begin{eqnarray}
4(v_{xt}v-v_{x}v_{t})-v^4+1
=vS^{(1,-2)}-S^{(1,-1)}S^{(-2,0)}+\S{-1,0}.
\end{eqnarray}
The r.h.s. is nothing but the recurrence relation \eqref{Sij-k=1-} sitting on the point $(i,j)=(1,0)$,
 which then vanishes.
The reminded equation reads
\begin{eqnarray}
4(v_{xt}v-v_{x}v_{t})=v^4-1.
\label{sg-v}
\end{eqnarray}
By the transformation
\begin{equation}
v=e^{\frac{i\phi}{2}}, ~~\mathrm{i.e.}~\phi =-2i\ln{v},
\end{equation}
the above equation is transformed into the sine-Gordon equation
\begin{eqnarray}
\phi_{xt}=\sin{\phi},\label{SG4}
\end{eqnarray}
of which the solution is given by
\begin{eqnarray}
\phi=-2i\ln{(S^{(-1,0)}-1)}.
\end{eqnarray}

\section{Link to the known results}\label{sec-5}

It is well known that in the frame of the famous Sato Theory \cite{MJD-book-2000,OSTT-PTPS-1988},
solutions of integrable equations are provided through the so-called  $\tau$ functions
which are expressed either in Hirota's exponential polynomial form or in terms of Wronskians.
It is also well known that by the famous IST solutions of integrable systems are given through the
GLM integral equation.
In our paper by $\S{i,j}$ we defined functions $u, v,z$ and derived the potential KdV equation,  potential mKdV equation,
SKdV equation and  sine-Gordon equation. At meantime solutions of these equations are given through the
expression of $\S{i,j}$, i.e.
\begin{equation}
S^{(i,j)}=\bs^{\st}\,\bK^j(\bI+\bM)^{-1}\bK^i\br.
\label{Sij-3}
\end{equation}
In the following let us discuss connection of the above $\S{i,j}$ with $\tau$ functions and the trace expression in the operator method.

\subsection{Connection with $\tau$ function}

\begin{Proposition}\label{prop-5}
For the scalar function $\S{i,j}$ defined in \eqref{Sij-3} where $\bK, \bM, \br, \bs$ are formulated by the
Sylvester equation \eqref{SE-2} and $\br, \bs$ obey the evolution \eqref{evo-rs-x} in $x$-direction, we have
\begin{equation}
\S{0,0}=\frac{\tau_x}{\tau},
\label{Sij-tau}
\end{equation}
where
\begin{equation}
\tau=|\bI+\bM|.
\label{tau-M}
\end{equation}
\end{Proposition}

\begin{proof}
Noting that $\S{i,j}$ is a scalar function, making use of trace we write $\S{0,0}$ as
\[\S{0,0}=\bs^{\st} (\bI+\bM)^{-1} \br
=\mathrm{Tr} (\bs^{\st} (\bI+\bM)^{-1} \br)
=\mathrm{Tr} ( \br\bs^{\st} (\bI+\bM)^{-1}).
\]
According to \eqref{evo-Mx} we have
\[\br\bs^{\st}=\bM_x=(\bI+\bM)_x,\]
and consequently,
\begin{equation}
\S{0,0}=\mathrm{Tr} ( (\bI+\bM)_x (\bI+\bM)^{-1})=\frac{|\bI+\bM|_x}{|\bI+\bM|}.
\end{equation}
For the last step of the above equation we have employed a known result
\[\mathrm{Tr} ( \bA_x \bA^{-1})=\frac{|\bA|_x}{|\bA|}.\]
(See theorem 7.3 in Ref.\cite{Coddington-book}.)
Besides, $\bA^{-1}\bA_x$ is used to define logarithmic derivative of $\bA$ in operator algebra \cite{M-book-1987}.
\end{proof}

For the $\S{i,j}$ with arbitrary $i,j\in \mathbb{Z}$, one can have
\begin{equation}
S^{(i,j)}=\bs^{\st}\,\bK^j(\bI+\bM)^{-1}\bK^i\br
=-\,\frac{\left|\begin{array}{cc}
                 0 & \bs^{\st}\,\bK^j\\
                 \bK^i\br & \bI+\bM
                \end{array} \right| }{|\bI+\bM|}.
\label{Sij-4}
\end{equation}
In fact, it is not difficult to see that the numerator can be equally written as
\[\left|\begin{array}{cc}
                 -\bs^{\st}\,\bK^j(\bI+\bM)^{-1}\bK^i\br & 0 \\
                 \bK^i\br & \bI+\bM
                \end{array} \right|= - \S{i,j} |\bI+\bM|.\]
Thus, one can always write $\S{i,j}$ in the form
\begin{equation}
\S{i,j}=\frac{g}{\tau}
\label{Sij-g-tau}
\end{equation}
with some function $g$.
This copes with the rational expression of solutions of integrable systems.
In fact, both \eqref{Sij-tau} and \eqref{Sij-g-tau} are very often used in bilinearization of nonlinear equations.

\subsection{Connection with the operator method}

The operator method is based on Marchenko's work \cite{M-book-1987},
first proposed by Aden and Carl \cite{AC-1996-JMP},
and developed by Schiebold and her collaborators.
The method was reviewed in  \cite{CS-Non-1999,CS-DMV-2000}
and the connection with Cauchy-type matrices determined by the Sylvester equations
was investigated by Schiebold in \cite{S-LAA-2010}.

Actually, Marchenko's method, the operator method developed by Jena's group (Carl, Schiebold, et.al) and the
Cauchy matrix method proposed by Nijhoff, et.al, are closely related each other.
All these methods are direct and need direct verifications of solutions.
Taking the (potential) KdV equation as an example, in the operator method
the solution is defined through the  trace of the logarithmic derivative, i.e.
\begin{equation}
u=\mathrm{Tr}((\bI+\bM)^{-1}(\bI+\bM)_{x}),
\end{equation}
while in the Cauchy matrix method the solution is given by
\begin{equation}
u=\S{i,j}|_{(i,j)=(0,0)}=\S{0,0}=\bs^{\st}(\bI+\bM)^{-1}\br.
\end{equation}
We have already seen from the above subsection that
both expressions are the same.
With regard to the dispersion relations,
in the operator method the dispersion relation is defined on $\bM$,
e.g., $\bM_x=\bK\bM,~\bM_t=\bK^3 \bM$, (see \cite{AC-1996-JMP}),
while in the Cauchy matrix approach the dispersion is defined through $\br, \bs$, i.e., \eqref{evo-rs}.
Noting that $\bM$ can be factorized as $\bM=\bF\bG\bH$ where
$\bF, \bH$ and $\bG$ are respectively related to $\br, \bs$ and $\bK$, (see Appendix \ref{A:2}),
defining dispersion relations on $\br$ and $\bs$ will provide more freedom to analyze the evolution of $\S{i,j}$.
Actually, $\S{i,j}$ can evolve w.r.t. the discrete coordinates $(i,j)$ and continuous ones $x$ and $t$.
$\S{i,j}$ defined on some adjacent points may form closed forms which can be viewed as discrete equations.
These discrete equations used to play  indispensable roles in our verification.

\subsection{Connection with direct linearization approach}

In 1981, Fokas and Ablowitz provided a linearization approach which allows one to obtain a larger class of solutions for
the KdV equation \cite{FA-PRL-1981}. The approach starts from an integral equation
\begin{equation}
\phi(k;t,x)+i e^{i(kx-k^3 t)}\oint_C\frac{\phi(z;t,x)}{k+z}d\lambda(z)=e^{i(kx-k^3t)},
\end{equation}
where $d\lambda(z)$ and $C$ are an appropriate measure and contour, respectively.
If $\phi(k;t,x)$ solves the above equation, then
\begin{equation}
w=-\partial_x\oint_C\phi(k;t,x)d\lambda(k)
\end{equation}
satisfies the KdV equation
\[w_t-6ww_x-w_{xxx}=0.\]

It is hard to find a relation between the linearization approach and our approach based on the Sylvester equation.
However, such a relation is visible in the fully discrete case.
Nijhoff and his collaborators invented discrete versions of direct linearization approach\cite{NQC-PLA-1983} and
the Cauchy matrix approach\cite{NAH-2009-JPA}.
Tutorial procedures of these two approaches can be found in \cite{N-arxiv-2011} and \cite{NAH-2009-JPA}.
Making a comparison of the two procedures, one can see many correspondences.
Here the infinity matrix $\bS$ is nothing but the main matrix $\mathbf{U}$ in direct linearization approach.

\subsection{Deformation of the Sylvester equation and rank one condition}

The Sylvester equation \eqref{SE} can be alternatively expressed, for example, as the following,
\begin{equation}
\bA \bX \bB-\bX=\bC,
\end{equation}
which was used to generate Toda lattice and 2-dimensional Toda lattice, (see Table 1 in \cite{S-LAA-2010}).

Note that the Sylvester equation \eqref{SE-2} (as well as the generic one \eqref{SE}) can be solved by factorizing 
$\bM=\bF\bG\bH$, see Appendix \ref{A:2}.
Since $\br$ and $\bs$ can be expressed via $\bF$ and $\bH$ (see \eqref{rFsH}),
and making use of the property \eqref{commut-2}, one can remove $\bF$ and $\bH$ form the Sylvester equation \eqref{SE-3}
and the remaining reads
\begin{equation}
 \Ga \bG+ \bG\Ga^T=\bI_{\ty{G}}\, \bI_{\ty{G}}^T,
\label{SE-const1}
\end{equation}
whit $\bI_{\ty{G}}$  given in \eqref{IG}.
This means it is also possible to start from the following 
Sylvester equation 
\begin{equation}
 \bX \bY- \bZ\bX=\bC,
\label{SE-const2}
\end{equation}
where all the elements are constant square matrices,
and one need to find a suitable way to insert evolution of $x,t$ or dispersion relations.
The following constant Sylvester-type equation 
\begin{equation}
 \bX \bY- \bZ\bX+\bI=\bC
\label{SE-const3}
\end{equation}
is fundamentally related to bispectral differential operators \cite{W-IM-1998}
and was used to constructed $\tau$-function of the KP hierarchy \cite{KG-JMP-2001,GK-TMP-2002,GK-JGP-2006,BGK-MPAG-2009}

As we mentioned before, the matrix $\bC$ in the Sylvester equation \eqref{SE} is required to be of rank one in the operator method
so as to get needed trace property. Actually, the rank one condition appeared in many literatures 
(e.g. \cite{AC-1996-JMP,S-PD-1998,CS-Non-1999,CS-DMV-2000,S-LAA-2010,B-Non-2000,HL-PLA-2001,W-IM-1998,KG-JMP-2001,GK-TMP-2002,GK-JGP-2006}).
Note that any $N$-th order rank one matrix is similar to the product $\br \bs^{\st}$ with suitable $N$-th order vectors $\br$ and $\bs$.
Recalling our master function $\S{i,j}=\bs^{\st}\, \bK^j (\bI+\bM)^{-1}\bK^i \br$, 
the rank one condition naturally guarantees (or comes from the fact) $\S{i,j}$ defines a scaler function.
When the vectors $\br$ and $\bs$ are replaced with matrices, the rank of the  product $\br \bs^{\st}$ 
will be higher than one and this will lead to matrix equations or noncommutative equations.

\section{Conclusions and discussions}\label{sec-6}

We have shown the links between the Sylvester equation and some continuous integrable systems.
The Sylvester equation of our interest is \eqref{SE-2}, i.e.,
\begin{equation}
\bK \bM+\bM\bK=\br\, \bs^{\st},
\label{SE-6}
\end{equation}
and it defines the master function $\S{i,j}$ \eqref{Sij}.
The recurrence relation satisfied by  $\S{i,j}$  (e.g. \eqref{Sij-k=1+} and \eqref{Sij-k=1-})
can be viewed as discrete equations of $\S{i,j}$ with discrete independent variables $i,j$.
After imposing dispersion relation on $\br$ and $\bs$,
we got evolution ($x$-, $t$-derivatives) of $\S{i,j}$.
Then we were able to derive several continuous integrable equations, including the KdV equation,
mKdV equation, SKdV equation and sine-Gordon equation.
These continuous equations cope with the continuous limit of their discrete counterparts (cf. \cite{NC-AAM-1995}).
The procedure can be viewed as a continuous version of the Cauchy matrix approach in the discrete case \cite{NAH-2009-JPA,ZZ-SAM-2013}.
The verification looks more natural than those in the operator method.

We finish the paper by the following remarks.

First of all, the main purpose of the paper is to display deep links between the Sylvester equation and integrable systems,
as well as to unify the Cauchy matrix approach in both discrete and continuous case.
The solution $\bM$ of the Sylvester equation directly leads to $\tau$ function through $\tau=|\bI+\bM|$.
Since the Sylvester equation is widely used in variety of areas such as control theory, signal processing,
and so forth, we believe there will more links to be found between integrable systems and these areas.

What's more, the master function $\S{i,j}$ shares the same definition in both discrete and continuous cases.
It is interesting that all the continuous equations considered in this paper
are reduced to discrete equations \eqref{Sij-k=1+} or \eqref{Sij-k=1-} sitting on some special points.
We also note that the relation \eqref{Sij-k=1+}  first appeared in \cite{NAH-2009-JPA}.
It does not play any role in the construction of lattice equations,
but it is really indispensable in continuous case.

In addition, the Sylvester equation \eqref{SE-6} can be exactly solved by factorizing $\bM$ into $\bF\bG\bH$.
For the detailed solving procedure one can refer to \cite{ZZ-SAM-2013,S-LAA-2010} and here we list out the main results in Appendix \ref{A:2}.
We made use of some special matrices and their properties (see Appendix \ref{A:1})
to prove the symmetric property $\S{i,j}=\S{j,i}$.

Besides, it is possible to generalize the Sylvester equation \eqref{SE} so that it admits elliptic functions as solutions \cite{N-pri-2013}.
This is being considered elsewhere.
Actually, there are already examples of direct linearization approach \cite{NP-JNMP-2003,JN-2013}
and of direct verification using dressed Cauchy matrices \cite{YKN-JMP-2013}.
Another extension is to replace vectors $\br$ and $\bs$ with matrices.
This will lead to matrix equations or noncommutative equations (see \cite{NZZ-2013} an example)
together with their solutions.
This extension corresponds to noncommutative operator-valued soliton equations \cite{CS-JMP-2009,S-DCDS-2009,CS-JMP-2011,S-JMP-2011,CS-JNMP-2012}.

Finally, we note that this paper will be followed by part II \cite{ZZ-2014}, where the relations between the Sylvester equation of the generic form \eqref{SE}
and coupled continuous integrable equations, semi-discrete integrable equations and higher dimensional
integrable systems will be investigated.
We are also interested in the discrete recurrence relation of $\S{i,j}$ and the continuous or semi-discrete equations that
are constructed by $\S{i,j}$.

\vskip 20pt
\subsection*{Acknowledgments}
The authors thank Prof. Tuncay Aktosun for the warm discussion.
Our thanks are also extended to  Dr. Ying-ying Sun for verifying the formulas of the paper
and to Mr. Wei Fu who wrote a {\it Mathematica} program to much simplify the calculations. 
This project is supported by the National Natural Science Foundation
Grant (Nos.11071157, 11371241, 11301483)
and the Project of ``First-class Discipline of Universities in Shanghai''.

\vskip 20pt

\begin{appendix}

\section{Some special matrices and properties} \label{A:1}

We will use some special matrices and their properties. Let us list them below.

\begin{itemize}
\item{Diagonal matrix:
\begin{equation}
\Ga^{\tyb{N}}_{\ty{D}}(\{k_j\}^{N}_{1})=\mathrm{Diag}(k_1, k_2, \ldots, k_N),
\end{equation}
}
\item{Jordan block matrix:
\begin{equation}
\Ga^{\tyb{N}}_{\ty{J}}(a)
=\left(\begin{array}{cccccc}
a & 0    & 0   & \cdots & 0   & 0 \\
1   & a  & 0   & \cdots & 0   & 0 \\
0   & 1  & a   & \cdots & 0   & 0 \\
\vdots &\vdots &\vdots &\vdots &\vdots &\vdots \\
0   & 0    & 0   & \cdots & 1   & a
\end{array}\right),
\end{equation}
}
\item{Lower triangular Toeplitz matrix:\footnote{More properties of this kind of matrices can be found in \cite{Zhang-KdV-2006}}
\begin{equation}
\bT^{\tyb{N}}(\{a_j\}^{N}_{1})
=\left(\begin{array}{cccccc}
a_1 & 0    & 0   & \cdots & 0   & 0 \\
a_2 & a_1  & 0   & \cdots & 0   & 0 \\
a_3 & a_2  & a_1 & \cdots & 0   & 0 \\
\vdots &\vdots &\cdots &\vdots &\vdots &\vdots \\
a_{N} & a_{N-1} & a_{N-2}  & \cdots &  a_2   & a_1
\end{array}\right)_{N\times N},
\label{T}
\end{equation}
}
\item{Skew triangular Toeplitz matrix:
\begin{equation}
\bH^{\tyb{N}}(\{b_j\}^{N}_{1})
=\left(\begin{array}{ccccc}
b_1 & \cdots  & b_{N-2}  & b_{N-1} & b_N\\
b_2 & \cdots & b_{N-1}  & b_N & 0\\
b_3 &\cdots & b_N & 0 & 0\\
\vdots &\vdots & \vdots & \vdots & \vdots\\
b_N & \cdots & 0 & 0 & 0
\end{array}
\right)_{N\times N},
\label{H}
\end{equation}
}
\end{itemize}

The lower triangular Toeplitz matrices and skew triangular Toeplitz matrices defined above
have the following property.
\begin{Proposition}\label{prop-A-1}
Let
\begin{subequations}
\begin{align}
& \mathcal{T}^{\tyb{N}}=\{\bT^{\tyb{N}}(\{a_j\}^{N}_{1})\},\\
& \bar{\mathcal{T}}^{\tyb{N}}=\{\bH^{\tyb{N}}(\{b_j\}^{N}_{1})\}.
\end{align}
\end{subequations}
Then we have\\
\textrm{(1).}~ $\bA \bB=\bB \bA,~\forall \bA, \bB\in \mathcal{T}^{\tyb{N}}$;\\
\textrm{(2).}~ $\bH=\bH^T,~\forall \bH\in \bar{\mathcal{T}}^{\tyb{N}}$;\\
\textrm{(3).}~ $\bH \bA=(\bH \bA)^T=\bA^T \bH,~\forall \bA\in \mathcal{T}^{\tyb{N}},~\forall \bH\in \bar{\mathcal{T}}^{\tyb{N}}$.
\end{Proposition}

This proposition can be extended to the following generic case.
\begin{Proposition}\label{prop-A-2}
Let
\begin{subequations}
\begin{align}
& \mathcal{G}^{\tyb{N}}=\{\mathrm{Diag}(\Ga^{\tyb{N}}_{\ty{D}}(\{a_{1,j}\}^{N_1}_{1}), \bT^{\tyb{N}}(\{a_{2,j}\}^{N_2}_{1}),
\bT^{\tyb{N}}(\{a_{3,j}\}^{N_3}_{1}), \cdots , \bT^{\tyb{N}}(\{a_{s,j}\}^{N_s}_{1}))\},\\
& \bar{\mathcal{G}}^{\tyb{N}}=\{\mathrm{Diag}(\Ga^{\tyb{N}}_{\ty{D}}(\{b_{1,j}\}^{N_1}_{1}), \bH^{\tyb{N}}(\{b_{2,j}\}^{N_2}_{1}),
\bH^{\tyb{N}}(\{b_{3,j}\}^{N_3}_{1}), \cdots , \bH^{\tyb{N}}(\{b_{s,j}\}^{N_s}_{1}))\},
\end{align}
where $0\leq N_j\leq N$ for $j=0,1,\cdots,N$ and $\sum^{s}_{j=1}N_j=N$.
\end{subequations}
Then we have\\
\textrm{(1).}~ $\bA \bB=\bB \bA,~\forall \bA, \bB\in \mathcal{G}^{\tyb{N}}$;\\
\textrm{(2).}~ $\bH=\bH^T,~\forall \bH\in \bar{\mathcal{G}}^{\tyb{N}}$;\\
\textrm{(3).}~ $\bH \bA=(\bH \bA)^T=\bA^T \bH,~\forall \bA\in \mathcal{G}^{\tyb{N}},~\forall \bH\in \bar{\mathcal{G}}^{\tyb{N}}$.
\end{Proposition}

\section{Solutions to the Sylvester equation \eqref{SE-2}} \label{A:2}

Since $\S{i,j}$ is invariant under any similar transformation of $\bK$, (see Sec.\ref{sec:2.2.2})
here we give the solution of the the Sylvester equation of the following form,
\begin{equation}
 \Ga \bM+ \bM\Ga=\br \bs^{\st},
\label{SE-3}
\end{equation}
where $\Ga$ is the canonical form of $\bK$.
We note that this equation is already solved in \cite{S-LAA-2010} and \cite{ZZ-SAM-2013}.
For convenience we follow the notations used in \cite{ZZ-SAM-2013}.
Besides those matrix notations we used in the previous section,
let us list  other needed ones below.
\begin{subequations}\label{notations}
\begin{align}
& N\mathrm{\hbox{-}th~order~vector:}~~ \br=(r_1, r_2, \cdots, r_N)^T,~~  \bs=(s_1, s_2, \cdots, s_N)^T,\\
& N\mathrm{\hbox{-}th~order~vector:}~~\bI_{\ty{D}}^{\tyb{N}}=(1,1,1, \ldots, 1)^T,~~ \bI_{\ty{J}}^{\tyb{N}}=(1,0,0, \ldots, 0)^T,\\
& N\times N ~\mathrm{matrix:}~~\bG^{\tyb{N}}_{\ty{D}}(\{k_j\}^{N}_{1})
=(g_{i,j})_{N\times N},~~~g_{i,j}=\frac{1}{k_i+k_j},\\
& N_1\times N_2 ~\mathrm{matrix:}~~\bG^{\tyb{N$_1$,N$_2$}}_{\ty{DJ}}(\{k_j\}^{N_1}_{1};a)
=(g_{i,j})_{N_1\times N_2},~~~g_{i,j}=-\Bigl(\frac{-1}{k_i+a}\Bigr)^j,\\
& N_1\times N_2 ~\mathrm{matrix:}~~\bG^{\tyb{N$_1$,N$_2$}}_{\ty{JJ}}(a;b)
=(g_{i,j})_{N_1\times N_2},~~~g_{i,j}=\mathrm{C}^{i-1}_{i+j-2}\frac{(-1)^{i+j}}{(a+b)^{i+j-1}},\\
& N\times N ~\mathrm{matrix:}~~\bG^{\tyb{N}}_{\ty{J}}(a)=\bG^{\tyb{N,N}}_{\ty{JJ}}(a;a)
=(g_{i,j})_{N\times N},~~~g_{i,j}=\mathrm{C}^{i-1}_{i+j-2}\frac{(-1)^{i+j}}{(2a)^{i+j-1}},\label{G-J-a}
\end{align}
\end{subequations}
where
\[\mathrm{C}^{i}_{j}=\frac{j!}{i!(j-i)!},~~(j\geq i).\]

The procedure for solving the Sylvester equation \eqref{SE-3} can be found in \cite{ZZ-SAM-2013,S-LAA-2010}.
The key point in the solving procedure is to factorize $\bM$ in to a triplet, i.e. $\bM  =\bF \bG \bH$.
Here we skip the detailed procedure and directly write out solutions.
\begin{Proposition}\label{prop-A-3}
(1).~ When
\begin{equation}
\Ga=\Ga^{\tyb{N}}_{\ty{D}}(\{k_j\}^{N}_{1}),
\label{Ga-D}
\end{equation}
we have
\begin{subequations}
\begin{equation}
\bM  =\bF \bG \bH =\Bigl(\frac{r_i s_j}{k_i+k_j}\Bigr)_{N\times N}\, ,
\end{equation}
where
\begin{equation}
\bF=\Ga^{\tyb{N}}_{\ty{D}}(\{r_j\}^{N}_{1}),~~
\bG=\bG^{\tyb{N}}_{\ty{D}}(\{k_j\}^{N}_{1}),~~
\bH=\Ga^{\tyb{N}}_{\ty{D}}(\{s_j\}^{N}_{1}).
\end{equation}
\end{subequations}

\vskip 5pt
\noindent
(2).~When
\begin{equation}
\Ga=\Ga^{\tyb{N}}_{\ty{J}}(k_1),
\label{Ga-J}
\end{equation}
we have
\begin{subequations}
\begin{equation}
\bM=\bF  \bG  \bH,
\end{equation}
where
\begin{equation}
\bF=\bT^{\tyb{N}}(\{r_j\}^{N}_{1}),~~\bG=\bG^{\tyb{N}}_{\ty{J}}(k_1),~~ \bH=\bH^{\tyb{N}}(\{s_j\}^{N}_{1}).
\end{equation}
\end{subequations}

\vskip 5pt
\noindent
(3).~When
\begin{equation}
\Ga=\mathrm{Diag}\bigl(\Ga^{\tyb{N$_1$}}_{\ty{D}}(\{k_j\}^{N_1}_{1}),
\Ga^{\tyb{N$_2$}}_{\ty{J}}(k_{N_1+1}),\Ga^{\tyb{N$_3$}}_{\ty{J}}(k_{N_1+2}),\cdots,
\Ga^{\tyb{N$_s$}}_{\ty{J}}(k_{N_1+(s-1)})\bigr),
\label{Ga-gen}
\end{equation}
where $\sum^{s}_{j=1}N_j=N$ (and we note that in this case $\Ga\in \mathcal{G}^{\tyb{N}}$),
we have
\begin{subequations}
\label{r-M-g}
\begin{equation}
\bM=\bF\bG \bH,
\end{equation}
where
\begin{align}
&\bF=\mathrm{Diag}\bigl(
\Ga^{\tyb{N$_1$}}_{\ty{D}}(\{r_j\}^{N_1}_{1}),
\bT^{\tyb{N$_2$}}(\{r_j\}^{N_1+N_2}_{N_1+1}),\bT^{\tyb{N$_3$}}(\{r_j\}^{N_1+N_2+N_3}_{N_1+N_2+1}),\cdots,
\bT^{\tyb{N$_s$}}(\{r_j\}^{N}_{1+\sum^{s-1}_{j=1}N_j})
\bigr), \\
&\bH=\mathrm{Diag}\bigl(
\Ga^{\tyb{N$_1$}}_{\ty{D}}(\{s_j\}^{N_1}_{1}),
\bH^{\tyb{N$_2$}}(\{s_j\}^{N_1+N_2}_{N_1+1}),\bH^{\tyb{N$_3$}}(\{s_j\}^{N_1+N_2+N_3}_{N_1+N_2+1}),\cdots,
\bH^{\tyb{N$_s$}}(\{s_j\}^{N}_{1+\sum^{s-1}_{j=1}N_j})
\bigr),
\end{align}
and $\bG$ is a symmetric matrix with block structure
\begin{equation}
\bG=\bG^T=(\bG_{i,j})_{s\times s}
\end{equation}
with
\begin{equation}
\begin{array}{ll}
\bG_{1,1}=\bG^{\tyb{N$_1$}}_{\ty{D}}(\{k_j\}^{N_1}_{1}),&~\\
\bG_{1,j}=\bG_{j,1}^T=\bG^{\tyb{N$_1$,N$_j$}}_{\ty{DJ}}(\{k_j\}^{N_1}_{1};k_{N_{j-1}+1}),&~~(1<j\leq s), \\
\bG_{i,j}=\bG_{j,i}^T=\bG^{\tyb{N$_i$,N$_j$}}_{\ty{JJ}}(k_{N_{i-1}+1};k_{N_{j-1}+1}),&~~(1<i\leq j\leq s).
\end{array}
\end{equation}
\end{subequations}
\end{Proposition}

\section{Proof of proposition \ref{prop-4}: $\S{i,j}=\S{j,i}$}\label{A:3}

The generic case of the Sylvester equation corresponds to $\Ga$ being defined in \eqref{Ga-gen}.
In this case,
\[\bM=\bF  \bG  \bH\]
where $\bF, \bH, \bG$ are given in the item (3) of proposition \ref{prop-A-3},
and $\bF, \bH$ and $\br, \bs$ are related through
\begin{equation}
\br=\bF\,\bI_{\ty{G}},~~~\bs= \bH\, {\bI_{\ty{G}}},
\label{rFsH}
\end{equation}
where
\begin{equation}
\bI_{\ty{G}}=\left(
\bI_{\ty{D}}^{\tyb{N$_1$}^T},
\bI_{\ty{J}}^{\tyb{N$_2$}^T},
\bI_{\ty{J}}^{\tyb{N$_3$}^T},
\ldots,
\bI_{\ty{J}}^{\tyb{N$_s$}^T}
\right)^T.
\label{IG}
\end{equation}
Noting that from proposition \ref{prop-A-3} we know that
\begin{equation}
 \bG=\bG^T,~~\bH=\bH^T,~~
 \bF\Ga=\Ga\bF,~~ \bH\Ga=\Ga^T\bH,~~ \bH \bF=(\bH\bF)^T,
\label{commut-2}
\end{equation}
we have
\begin{equation*}
\begin{array}{rl}
 S^{(i,j)} & =  \bs^{\st}\, \Ga^j(\bI+ \bM)^{-1}\Ga^i \br\\
 & =\bI^{T}_{\ty{G}}\, \bH \Ga^j(\bI+ \bF\bG\bH)^{-1}\Ga^i\bF\,\bI^{}_{\ty{G}}\\
 & =\bI^{T}_{\ty{G}}\, (\Ga^T)^j \bH (\bI+ \bF\bG\bH)^{-1}\bF\Ga^i\, \bI^{}_{\ty{G}}\\
 & = \bI^{T}_{\ty{G}}\, (\Ga^T)^j((\bH\bF)^{-1} + \bG)^{-1}\Ga^i\, \bI^{}_{\ty{G}}\\
 & = \bI^{T}_{\ty{G}}\, (\Ga^T)^i((\bH\bF)^{-1} + \bG)^{-1}\Ga^j\, \bI^{}_{\ty{G}}\\
 & =S^{(j,i)},
\end{array}
\end{equation*}
where we have made use of the fact $(\bH\bF)^{-1} + \bG =((\bH\bF)^{-1} + \bG)^T$.
This proves proposition \ref{prop-4}.

\section{Solutions for the obtained nonlinear equations}\label{A:4}

We already have solution $\bM$ for the Sylvester equation \eqref{SE-3} where $\bM$ is factorized as
$$\bM=\bF\bG\bH,$$
$\bF$ is related to $\br$,
$\bH$ is related to $\bs$
and $\bG$ is related to $\Ga$.
Once we find $\br$ and $\bs$ that satisfy suitable dispersion relations, we can have the explicit form for
\begin{equation}
\S{i,j}=\bs^{\st}\, \Ga^j (\bI+\bM)^{-1}\Ga^i \br,
\label{Sij-Ga}
\end{equation}
and then solutions of the corresponding nonlinear evolution equations.

Let us first consider the following $x$-evolution
\begin{equation}
\br_{x}=\Ga \br,~~ \bs_{x}= \Ga^T \bs
\label{evo-rs-x-Ga}
\end{equation}
and leave the $t$-evolution open.
We suppose
\begin{subequations}
\label{rho-sigma}
\begin{align}
&\rho_i=e^{\xi_i},~~\xi_i=k_{i}x+f(k_{i},t)+\xi^{(0)}_i,~\mathrm{with~ constants~}\xi^{(0)}_i,\\
&\sg_i=e^{\eta_i},~~\eta_i=k_{i}x+f(k_{i},t)+\eta^{(0)}_i,~\mathrm{with~ constants~}\eta^{(0)}_i,
\end{align}
where for this moment $f(k_{i},t)$ is some function of $(k_{i},t)$ and independent of $x$.
\end{subequations}
Then we have the following.\\
(1).~ When $\Ga$ is given by \eqref{Ga-D}, i.e.
\begin{equation}
\Ga=\Ga^{\tyb{N}}_{\ty{D}}(\{k_j\}^{N}_{1}),
\end{equation}
we have
\begin{subequations}
\begin{align}
&\br=\br_{\hbox{\tiny{\it D}}}^{\hbox{\tiny{[{\it N}]}}}(\{k_j\}_{1}^{N})=(r_1, r_2, \cdots, r_N)^T, ~~\mathrm{with}~ r_i=\rho_i,\\
&\bs=\bs_{\hbox{\tiny{\it D}}}^{\hbox{\tiny{[{\it N}]}}}(\{k_j\}_{1}^{N})=(s_1, s_2, \cdots, s_N)^T, ~~\mathrm{with}~ s_i=\sigma_i.
\end{align}
\end{subequations}
(2).~When $\Ga$ is given by \eqref{Ga-J}, i.e.
\begin{equation}
\Ga=\Ga^{\tyb{N}}_{\ty{J}}(k_1),
\end{equation}
we have
\begin{subequations}
\begin{align}
&\br=\br_{\hbox{\tiny{\it J}}}^{\hbox{\tiny{[{\it N}]}}}(k_1)=(r_1, r_2, \cdots, r_N)^T, ~~\mathrm{with}~ r_i=\frac{\partial^{i-1}_{k_1}\rho_1}{(i-1)!},\\
&\bs=\bs_{\hbox{\tiny{\it J}}}^{\hbox{\tiny{[{\it N}]}}}(k_1)=(s_1, s_2, \cdots, s_N)^T, ~~\mathrm{with}~ s_i=\frac{\partial^{N-i}_{k_1}\sigma_1}{(N-i)!}.
\end{align}
\end{subequations}
(3).~When $\Ga$ is given by \eqref{Ga-gen}, i.e.
\begin{equation}
\Ga=\mathrm{Diag}\bigl(\Ga^{\tyb{N$_1$}}_{\ty{D}}(\{k_j\}^{N_1}_{1}),
\Ga^{\tyb{N$_2$}}_{\ty{J}}(k_{N_1+1}),\Ga^{\tyb{N$_3$}}_{\ty{J}}(k_{N_1+2}),\cdots,
\Ga^{\tyb{N$_s$}}_{\ty{J}}(k_{N_1+(s-1)})\bigr),
\end{equation}
we have
\begin{equation}
\br=\left(
\begin{array}{l}
\br_{\ty{D}}^{\tyb{N$_1$}}(\{k_j\}_{1}^{N_1})\\
\br_{\ty{J}}^{\tyb{N$_2$}}(k_{N_1+1})\\
\br_{\ty{J}}^{\tyb{N$_3$}}(k_{N_1+2})\\
\vdots\\
\br_{\ty{J}}^{\tyb{N$_s$}}(k_{N_1+(s-1)})
\end{array}
\right),~~~\bs=\left(
\begin{array}{l}
\bs_{\ty{D}}^{\tyb{N$_1$}}(\{k_j\}_{1}^{N_1})\\
\bs_{\ty{J}}^{\tyb{N$_2$}}(k_{N_1+1})\\
\bs_{\ty{J}}^{\tyb{N$_3$}}(k_{N_1+2})\\
\vdots\\
\bs_{\ty{J}}^{\tyb{N$_s$}}(k_{N_1+(s-1)})
\end{array}
\right)^T.
\end{equation}

Finally, for the time evolution
\begin{equation}
\br_t=4 \bK^{3} \br,~~\bs_t=4 (\bK^T)^{3}\bs,
\end{equation}
one needs to take $f(k_i,t)=4k_i^3t$ in \eqref{rho-sigma}, and for
\begin{equation}
\br_t=\frac{1}{4} \bK^{-1} \br,~~\bs_t=\frac{1}{4} (\bK^T)^{-1}\bs,
\end{equation}
one takes $f(k_i,t)=\frac{t}{4k_i}$ in \eqref{rho-sigma}.

We note that it is well known that $\Ga$ being a Jordan block leads to multiple-pole solutions in IST \cite{WO-JPSJ-1982},
singular solutions or higer order positons or negatons (for the KdV equation) \cite{M-book-1987,M-PLA-1992-1,M-PLA-1992-2},
which can also be derived from soliton solutions through a limiting procedure (cf.\cite{Zhang-KdV-2006}).

\section{Higher order KdV equations}\label{A:5}

One can obtain higher order nonlinear evolution equations. Let us take the higer order KdV equations as an example.
Consider the following evolution,
\begin{subequations}
\begin{align}
& \br_x=\bK \br,~~\bs_x= \bK^T \bs,\\
& \br_{t_{2n+1}}=4^{n}\bK^{2n+1}\br, ~~\bs_{t_{2n+1}}=4^{n}(\bK^{2n+1})^T\bs.
\end{align}
\label{evo-rs-hie}
\end{subequations}
It then follows that
\begin{subequations}
\begin{eqnarray}
&& \bM_x=\br \bs^{\st}, \label{Mx-hie} \\
&& \bM_{t_{2n+1}}=4^n\sum_{l=0}^{2n}(-1)^l\bK^{2n-l}\br\bs^{\st}\bK^l,~~(n=1,2,3,\ldots), \label{Mt-hie}
\end{eqnarray}
\end{subequations}
and for $\S{i,j}$, besides those $x$-derivatives in \eqref{evo-Sij},
\begin{equation}
S^{(i, j)}_{t_{2n+1}}=4^n\left(S^{(i, j+2n+1)}+S^{(i+2n+1, j)}-\sum_{l=0}^{2n}(-1)^l S^{(2n-l, j)}S^{(i, l)}\right),~~(n=1,2,\cdots).
\end{equation}
We still take $u=S^{(0,0)}$ and find that
\begin{eqnarray}
u_{xt_{2n+1}}-u_{xxxt_{2n-1}}-8u_{x}u_{xt_{2n-1}}-4u_{xx}u_{t_{2n-1}}=0 \label{pKdV-hie}
\end{eqnarray}
holds for $n=1,2,3,4,5$.
We check the above relation by means of Mathematica, which is not complicated.
For each $n$, the l.h.s. of \eqref{pKdV-hie} reduces to the recurrence relation \eqref{Sij-k=1+} sitting on some points $(i,j)$.
In terms of $w=2\mu_x$, \eqref{pKdV-hie} is written as
\begin{subequations}
\begin{eqnarray}
w_{t_{2n+1}}=w_{xxt_{2n-1}}+4ww_{t_{2n-1}}+2w_{x}\partial_x^{-1}w_{t_{2n-1}}= {R}\,w_{t_{2n-1}} , \label{KdV-hie}
\end{eqnarray}
where
\begin{equation}
R=\partial_{xx}+4w+2w_x\partial^{-1}_x
\end{equation}
\end{subequations}
is known as the recursive operator of the KdV hierarchy.

We note that for the matrix KdV hierarchy a recursion operator has been given and solutions were verified (cf. \cite{CS-JMP-2011}).

\end{appendix}

\end{document}